\documentclass[copyright,creativecommons,noderivs,11pt]{eptcs}
\providecommand{\event}{GandALF 2017} %

\usepackage[T1]{fontenc} %
\usepackage{amsmath,amsthm}
\usepackage[cmintegrals]{newtxmath}

\usepackage{bm}

\usepackage[utf8]{inputenc}

\usepackage{ellipsis, mparhack, ragged2e} %
\usepackage[l2tabu, orthodox]{nag} %

\usepackage{graphicx}
\usepackage{subcaption}
\usepackage{amssymb,mathtools}
\usepackage{epstopdf}
\usepackage{url}
\usepackage{csquotes}
\usepackage{wasysym}
\usepackage{tabularx}
\usepackage{multirow}

\usepackage{float}
\usepackage{placeins}
\usepackage{todonotes}

\usepackage{tikz}
\usetikzlibrary{automata, arrows.meta, decorations.text, decorations.markings, positioning}

\usepackage{textcomp}
\usepackage{xspace}
\usepackage{xcolor}
\usepackage{booktabs}
\usepackage{colortbl}

\usepackage{cleveref}

\usepackage{pdflscape}

\usepackage{numprint}
\npdecimalsign{.}
\npthousandsep{,}

\newcommand{\pnodes}[1]{\nprounddigits{0}\numprint{#1}}
\newcommand{\psec}[1]{\nprounddigits{1}\npfourdigitnosep\numprint[s]{#1}}

\tikzset{
    state/.style={
		        rectangle,
            rounded corners,
            draw=black,
            minimum height=2em,
            minimum width=2em,
            align=center,
            }
}

\tikzset{
    statep/.style={
            circle,
            draw=black,
            minimum height=2em,
            minimum width=2em,
            align=center,
            }
}

\tikzstyle{acc}=[double]

\newcommand{\true}{{\ensuremath{\mathbf{tt}}}}
\newcommand{\false}{{\ensuremath{\mathbf{ff}}}}
\newcommand{\F}{{\ensuremath{\mathbf{F}}}}
\newcommand{\G}{{\ensuremath{\mathbf{G}}}}
\newcommand{\U}{{\ensuremath{\mathbf{U}}}}

\newcommand{\R}{{\ensuremath{\mathbf{R}}}}
\newcommand{\X}{{\ensuremath{\mathbf{X}}}}

\newcommand{\aft}{{\it af}}

\newcommand{\cnf}[1]{\mathsf{cnf}(#1)}
\newcommand{\dnf}[1]{\mathsf{dnf}(#1)}
\newcommand{\sff}[1]{\mathsf{sf}(#1)}
\newcommand{\conj}[1]{\mathsf{conj}(#1)}
\newcommand{\disj}[1]{\mathsf{disj}(#1)}
\newcommand{\support}[1]{\mathsf{support}(#1)}

\newcommand{\history}[1]{\ensuremath{\mathcal{H}(#1)}}
\newcommand{\historyunion}{\ensuremath{\sqcup}}
\newcommand{\historyinter}{\ensuremath{\sqcap}}
\newcommand{\automaton}[1]{\ensuremath{\mathcal{A}(#1)}}
\newcommand{\prodautomaton}[1]{\ensuremath{\mathcal{A}^\times(#1)}}
\newcommand{\cl}[1]{\ensuremath{\mathsf{cl}(#1)}}
\newcommand{\drop}[1]{\ensuremath{\mathsf{drop}(#1)}}
\newcommand{\gls}{\ensuremath{\mathsf{gls}}}

\newcommand{\qhold}{\ensuremath{q_{\mathsf{hold}}}}
\newcommand{\qacc}{\ensuremath{q_{\mathsf{acc}}}}
\newcommand{\qrej}{\ensuremath{q_{\mathsf{rej}}}}

\newcommand{\INF}[1]{\ensuremath{\mathit{Inf}(#1)}}
\newcommand{\FIN}[1]{\ensuremath{\mathit{Fin}(#1)}}

\newcommand{\ltltodstar}{\texttt{ltl2dstar}}
\newcommand{\ltltotgba}{\texttt{ltl2tgba}}
\newcommand{\spot}{\texttt{SPOT}}
\newcommand{\rabinizer}{\texttt{Rabinizer}}
\newcommand{\delag}{\texttt{Delag}}
\newcommand{\prism}{\texttt{PRISM}}

\newcommand{\pmin}[1]{\mathbf{P}_{\mathsf{min}}\left(#1\right)}
\newcommand{\pmax}[1]{\mathbf{P}_{\mathsf{max}}\left(#1\right)}

\colorlet{salomonColor}{cyan}
\colorlet{davidColor}{lime}

\newcommand{\rectangleWidth}{0.026 \textwidth}

\newcommand{\lift}[1]{\ensuremath{\uparrow#1}}

\newtheorem{definition}{Definition}
\newtheorem{theorem}{Theorem}
\newtheorem{lemma}{Lemma}

\title{LTL to Deterministic Emerson-Lei Automata}
\def\titlerunning{LTL to Deterministic Emerson-Lei Automata}
\def\authorrunning{D. M\"uller  \& S. Sickert}
\author{David M\"uller 
\institute{Technische Universität Dresden}
\thanks{This work is funded by the DFG-project BA-1679/12-1 and partially funded
    by the DFG Research Training Group \enquote{QuantLA: Quantitative Logics and
    Automata} (GRK 1763)}
\email{david.mueller2@tu-dresden.de}
\and Salomon Sickert
\institute{Technische Universität München}
\thanks{This work is funded by the DFG Research Training Group \enquote{PUMA: Programm- und Modell-Analyse} (GRK 1480)}
\email{sickert@in.tum.de}}
\begin{document}

\maketitle

\begin{abstract}
We introduce a new translation from linear temporal logic (LTL) to deterministic
Emerson-Lei automata, which are $\omega$-automata with a Muller acceptance
condition symbolically expressed as a Boolean formula. The richer
acceptance condition structure allows the shift of complexity from the state
space to the acceptance condition. Conceptually the construction is an enhanced
product construction that exploits knowledge of its components to reduce the
number of states. We identify two fragments of LTL, for which one can easily
construct deterministic automata and show how knowledge of these components can
reduce the number of states. We extend this idea to a
general LTL framework, where we can use arbitrary LTL to deterministic automata
translators for parts of formulas outside the mentioned fragments. Further, we
show succinctness of the translation compared to existing construction. The
construction is implemented in the tool \delag, which we evaluate 
on several benchmarks of LTL formulas and probabilistic model checking case studies.
\end{abstract}

\vspace{-0.8em}

\section{Introduction}

Deterministic $\omega$-automata play an essential role in the verification of probabilistic systems and in the synthesis of reactive systems, which generally prohibit a direct use of non-deterministic automata. However, determinisation of non-deterministic automata may cause an exponential blow-up, which makes these applications computationally hard. Hence there exists a long line of research aiming at shrinking the size of the generated deterministic automata as far as possible. All these translations have in common that they target a specific acceptance condition, such as Rabin, Streett, or Parity, and thus have to sometimes store progress information of the acceptance condition in the state.

In this work, we reexamine the Muller acceptance condition with a crucial twist: Instead of an explicit representation, we represent the acceptance condition in a symbolic fashion, as presented in \cite{Hanoi-CAV15}, which we call Emerson-Lei acceptance. Moving to a compactly expressed acceptance condition allows us to reduce the number of states and to use fewer acceptance sets compared to existing translations, although there is a well-known exponential lower bound for the size of deterministic $\omega$-automata starting from a non-deterministic $\omega$-automaton \cite{SV89}. Of course algorithms need to be adapted to this more complex scenario, but we present examples where this reduces the time needed for probabilistic model checking.

\paragraph{Related Work.}

There are two lines of research to cope with the exponential blow-up caused by determinisation of $\omega$-automata. The first explores restricted forms of non-determinism that are still usable for probabilistic verification, such as limit-deterministic automata \cite{Vardi85,DBLP:journals/jacm/CourcoubetisY95,SEJK16,SK16} or good-for-games-automata \cite{HP06, KMBK14} for Markov decision processes, or unambiguous B\"uchi automata for Markov chains \cite{BKKKMW16}. The authors of \cite{HLST015} try to avoid the full Safra's determinisation by under-approximating and over-approximating it via break-point and powerset construction. In the context of synthesis, one can evade determinisation using universal co-Büchi tree automata instead of deterministic parity automata \cite{KPV06}.

The second line of research aims at reducing the size of the state space of the resulting deterministic automaton. The most prominent determinisation method, Safra's determinisation, translates a non-deterministic B\"uchi into a deterministic Rabin automaton \cite{Safra88}. This translation is implemented in \ltltodstar{} with several heuristics \cite{KB06,KB07}. In the last decades there has been a lot of progress on determinisation of Büchi automata refining Safra's construction \cite{Piterman07,KW08,MS08,Schewe09,DBLP:journals/fuin/Redziejowski12,FL15}. While there still remains the exponential lower bound, efficient implementations are also available in \spot{} \cite{DBLP:conf/atva/Duret-LutzLFMRX16}. There has been also work on direct translations starting with fragments or even full LTL, see the history of \rabinizer{} \cite{BBKS13, EsparzaKS16}. The approach of \cite{SEJK16} originates from the same family of translations, which together with \cite{EKRS17}, yields an asymptotically optimal translation from LTL (via limit-deterministic automata) to Parity automata, which is implemented in \texttt{ltl2dpa}. The authors of \cite{MS10} follow a compositional approach where the LTL formula is brought into a normal form, decomposed, and then subformulas are translated separately. However, all these constructions target a specific acceptance condition structure --- Rabin, Streett, or Parity --- and thus sometimes need to encode the progress of the acceptance condition in the state space of the resulting automaton.
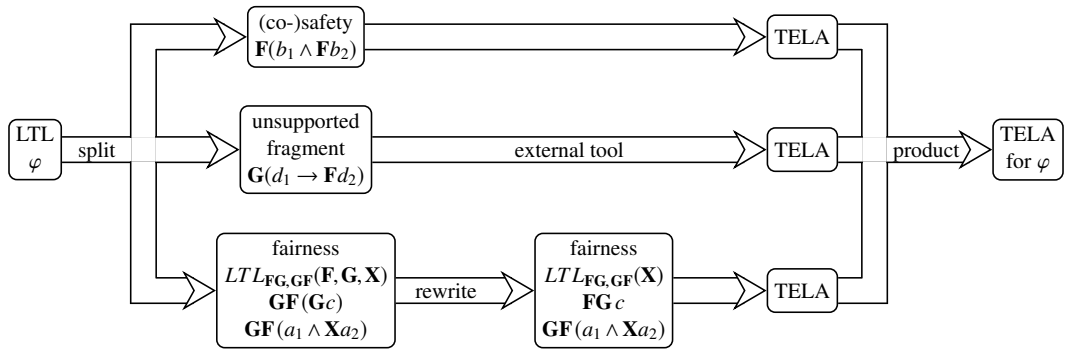
\begin{figure}[tbp] 
\centering
\scalebox{0.75}{
\begin{tikzpicture}
    \node[state,align=center,thick]   (ltl)   at (0,0)   {LTL\\\(\varphi\)};
    
    \node[rectangle,align=center, minimum width=\rectangleWidth, minimum height=\rectangleWidth,draw=white, line width=0]           (split)       at (0.12 \textwidth, 0)               {};

    \node[state,align=center,thick]   (safety)   at(0.3 \textwidth, 2) {(co-)safety\\\(\F (b_1 \wedge \F b_2)\)};
    \node[state,align=center,thick]   (tela1)    at(0.85 \textwidth, 2)     {TELA};

    \node[state,align=center,thick]   (fairness1)    at(0.3 \textwidth, -2.5) {fairness\\\(LTL_{\F\G, \G\F}(\F, \G, \X)\)\\\(\G \F \, (\G c)\)\\\(\G \F \, (a_1 \wedge \X a_2)\)};
    \node[state,align=center,thick]   (fairness2)   at (0.63 \textwidth, -2.5) {fairness\\\(LTL_{\F\G, \G\F}(\X)\)\\\(\F \G \, c\)\\\(\G \F \, (a_1 \wedge \X a_2)\)};
    \node[state,align=center,thick]   (tela2)       at (0.85 \textwidth, -2.5)    {TELA};

    \node[state,align=center,thick]   (unsupported)   at(0.3 \textwidth, 0) {unsupported\\fragment\\\(\G (d_1 \rightarrow \F d_2)\)};
    \node[state,align=center,thick]   (tela3)         at(0.85 \textwidth, 0)  {TELA};

    \node[rectangle,align=center, minimum width=\rectangleWidth, minimum height=\rectangleWidth,draw=white, line width=0]    (product)     at(0.93 \textwidth, 0)      {};

    \node[state,align=center,thick]   (prodTela)    at(1.1\textwidth, 0)  {TELA\\for \(\varphi\)};

    \draw[thick,-, double distance = 1.1 em]              (ltl) to (split);
    \draw[thick,-, decoration = {text along path, text={split}, text align = {align = center}, raise=-0.75 ex}, decorate]              (ltl) to (split);
    \draw[thick,-{Stealth[length=6mm, width=7mm, fill=white]}, double distance = 1.1 em]             (split)       |-                                      (fairness1);
    \draw[thick, double distance = 1.1 em, -{Stealth[length=6mm, width=7mm, fill=white]}]             (fairness1)   to                (fairness2);
    \draw[thick, decoration = {text along path, text={rewrite}, text align = {align = left, left indent=2 ex}, raise=-0.75 ex}, decorate, -]             (fairness1)   to                (fairness2);
    \draw[thick, double distance = 1.1 em, -{Stealth[length=6mm, width=7mm, fill=white]}]             (fairness2)   to                                                (tela2);
    \draw[thick, double distance = 1.1 em, -]              (tela2)       -|                                      (product);

    \draw[thick, double distance = 1.1 em, -{Stealth[length=6mm, width=7mm, fill=white]}]             (split)       |-                                       (safety);
    \draw[thick, double distance = 1.1 em, -{Stealth[length=6mm, width=7mm, fill=white]}]             (safety)      to                                                (tela1); 
    \draw[thick, double distance = 1.1 em, -]              (tela1)       -|                                       (product);

    \draw[thick, double distance = 1.1 em, -{Stealth[length=6mm, width=7mm, fill=white]}]             (split)       to                                                  (unsupported);
    \draw[thick, double distance = 1.1 em, -{Stealth[length=6mm, width=7mm, fill=white]}]             (unsupported) to                       (tela3);
    \draw[thick, decoration = {text along path, text={external tool}, text align = {align = center}, raise=-0.75 ex}, decorate, -]             (unsupported) to                       (tela3);
    \draw[thick, double distance = 1.1 em, -]              (tela3)       to                                                  (product);

    \draw[thick, double distance = 1.1 em,-{Stealth[length=6mm, width=7mm, fill=white]}]             (product)     to                  (prodTela);
    \draw[thick, decoration = {text along path, text={product}, text align = {align = left, left indent=0.7 ex}, raise=-0.75 ex}, decorate,-]             (product)     to           (prodTela);
\end{tikzpicture}
} \caption{The input LTL formula is split up, each subformula is translated independently, and then a product automaton is constructed, as can be seen for the example \(\varphi = \G \F \, (a_1 \wedge \X \, a_2) \wedge \F \, (b_1 \wedge \F \, b_2) \wedge \G \F \, (\G c) \wedge \G \, (c_1 \rightarrow \F c_2)\).}
\label{fig:outline}
\end{figure}
\paragraph{Contribution.}

We present a translation from LTL to deterministic Emerson-Lei automata that trades a compact state space for a more complex acceptance condition structure.  There has been previously the idea of a product construction relying on known translations in \cite{Hanoi-CAV15} to obtain a more complex acceptance condition. Here, we give a direct translation of fragments of LTL without an intermediate step over non-deterministic automata. We consider special liveness properties in particular and give a translation based on buffers. For safety and cosafety LTL formulas we rely on the $\aft$ function \cite{EsparzaKS16,SEJK16} computing the left-derivative directly on the formula. Additionally, if we encounter a subformula not contained in our supported fragments for a direct translation, we rely on external tools for translation, and compose a deterministic automaton for the overall formula. A general scheme for our approach is depicted in \Cref{fig:outline}, which we implemented in the tool \delag{} (Deterministic Emerson-Lei Automata Generator).

We conducted several experiments to evaluate the practical impact of this
idea: At first we compared the size of the automata measured in state
space size as well as acceptance sizes for our tool and several other tools like
\spot{} and \rabinizer{}. Secondly, we performed a case study (IEEE 802.11 Wireless LAN
Handshaking protocol) and also compared it with \spot{} and \rabinizer{}. On
both sides, we could show the potential of \delag{}, i.e., allowing arbitrary
acceptance conditions to obtain smaller automata. The implementation and
additional material can be found at \cite{gandalf17web}. %

\section{Preliminaries}

\subsection{Linear Temporal Logic}

We consider standard linear temporal logic (LTL) with all negations pushed down to the propositions.

\begin{definition}[LTL]\label{ltl}
A formula of LTL in {\em negation normal form} over a finite set of atomic propositions $(Ap)$ is given by the syntax:
\begin{align*}
\varphi::= & \; \true \mid \false \mid a \mid \neg a \mid \varphi \wedge \varphi \mid \varphi \vee \varphi \mid \X \varphi \mid \varphi \U \varphi \mid \varphi \R \varphi \qquad \text{with} ~ a \in Ap
\end{align*}
\noindent Furthermore, we introduce the abbreviations: $\F \varphi := \true \U \varphi$, $\G \varphi := \false \R \varphi$. 
An $\omega$-word $w$ is an infinite sequence of sets of atomic propositions $w[0] w[1] w[2] \cdots$ and we denote the infinite suffix $w[i] w[i+1] \cdots$ by $w_{i}$. The satisfaction relation $\models$ between $\omega$-words and formulas is inductively defined as follows: %
\[\begin{array}[t]{lclclcl}
w \models \true & & w \not\models \false \\
w \models a & \mbox{ if{}f } & a \in w[0] \\
w \models \neg a & \mbox{ if{}f } & a \not \in w[0] \\
w \models \varphi \wedge \psi & \mbox{ if{}f } & w \models \varphi \text{ and } w \models \psi\\
w \models \varphi \vee \psi & \mbox{ if{}f } & w \models \varphi \text{ or } w \models \psi
\end{array}
\qquad
\begin{array}[t]{lclclcl}
w \models \X \varphi & \mbox{ if{}f } & w_1 \models \varphi\\
w \models \varphi \U \psi & \mbox{ if{}f } & \exists i. \, w_i \models \psi \text{ and } \forall j < i. w_j \models \varphi \\
w \models \varphi \R \psi & \mbox{ if{}f } & \forall i. \, w_i \models \psi \text{ or }  \\ 
& & \exists i. \,  \, w_i \models \varphi \text{ and } \forall j \leq i. w_j \models \psi
\end{array}\]
Two formulas $\varphi, \psi$ are called equivalent, denoted $\varphi \equiv \psi$, if $w \models \varphi \leftrightarrow w \models \psi$ for all $w \in (2^{Ap})^\omega$. $\sff{\varphi}$ is defined as the set of temporal subformulas ($\U, \R, \X$) not nested within the scope of another temporal operator, e.g., $\sff{(\F\G a) \vee \X b} = \{\F\G a, \X b\}$.
\end{definition}

\subsection{Fragments of LTL}

We study several syntactic fragments of LTL. Let us denote by $LTL(\mathcal{X})$ the syntactic restriction of LTL to the temporal operators of $\mathcal{X}$. Furthermore we allow to denote prefixes that are applied to all formulas by a subscript: $LTL_{X, Y}(\mathcal{X}) = \{X \varphi, Y \varphi \mid \varphi \in LTL(\mathcal{X})\}$. We now identify three (well-known) syntactic LTL fragments commonly used in system property specifications: 

\begin{itemize}
	\item safety: $LTL(\R, \X)$
	\item reachability (or cosafety): $LTL(\U, \X)$
	\item fairness: $LTL_{\F\G, \G\F}(\F, \G, \X)$
\end{itemize}

\noindent We now show that the last fragment can be simplified to formulas without nested $\F$ and $\G$:
 
\begin{theorem}[Fairness LTL Normal Form]
Let $\varphi$ be an $LTL_{\F\G, \G\F}(\F, \G, \X)$ formula. Then there exists an equivalent formula $\varphi' \equiv \varphi$ that is a boolean combination of formulas in $LTL_{\F\G, \G\F}(\X)$.
\end{theorem} 

\begin{proof}
Exhaustive application of the following folklore equivalence-preserving rewrite rules, described in \cite{EH00,SomBloem00,LiSFZ16}, brings every fairness LTL formula into the desired normal form:
\[\begin{array}{rclrcl}
\F\G (\F \varphi) & \mapsto & \G\F\varphi & \G\F (\F \varphi) & \mapsto & \G\F\varphi \\
\F\G (\G \varphi) & \mapsto & \F\G\varphi & \G\F (\G \varphi) & \mapsto & \F\G\varphi \\
\F\G (\X \varphi) & \mapsto & \F\G\varphi & \G\F (\X \varphi) & \mapsto & \G\F \varphi \\
\F\G (\varphi \wedge \psi) & \mapsto & \F\G \varphi \wedge \F\G \psi & \G\F (\varphi \vee \psi) & \mapsto & \G\F \varphi \vee \G\F \psi \\
\F\G (\varphi \vee \F\psi)  & \mapsto & \F\G \varphi \vee \G\F \psi & \G\F (\varphi \wedge \F\psi) & \mapsto & \G\F \varphi \wedge \G\F \psi \\
\F\G (\varphi \vee \G\psi) & \mapsto & \F\G \varphi \vee \F\G \psi & \G\F (\varphi \wedge \G\psi) & \mapsto & \G\F \varphi \wedge \F\G \psi \\
\varphi \not \in LTL(\X) \Rightarrow \F\G (\varphi) & \mapsto & \F\G (\cnf{\varphi}) & \varphi \not \in LTL(\X) \Rightarrow \G\F (\varphi) & \mapsto & \G\F(\dnf{\varphi}) \\
\end{array}\]
with $\cnf{\varphi}$ and $\dnf{\varphi}$ denoting the translation into conjunctive and disjunctive normal form. 
\end{proof}

This translation might cause an exponential blow-up in formula size due to the translation into conjunctive and disjunctive normal form. However, the construction for fairness LTL to deterministic automata we present is only dependent on the size of the alphabet and the nesting depth of the $\X$-operators, which are both unchanged (or even decreased) by the translation. Further from now on we assume all fairness LTL formulas are rewritten to this normal form.

Apart from the rules listed above, our implementation uses several well-known simplification rules to rewrite formulas outside of the fairness fragment to formulas within, e.g., $\G\F (\varphi \U \psi) \mapsto \G\F \psi$ and $\F\G (\varphi \U \psi) \mapsto \G\F \psi \wedge \F\G (\varphi \vee \psi)$.

\subsection{Deterministic Emerson-Lei Automata}

Emerson-Lei automata are Muller automata with their acceptance condition expressed as a
\textit{generic acceptance condition} (see \cite{Hanoi-CAV15}): Instead of representing every Muller set 
explicitly a symbolic representation is used. We will take as acceptance condition a
Boolean combination over the atomic propositions $\FIN{P}$ and $\INF{P}$ where
$P$ is an arbitrary subset of transitions of an
$\omega$-automaton $\mathcal{A}$. We denote the set of all generic acceptance
conditions by $\mathcal{C}_\delta$.

\begin{definition}[Deterministic Transition-Based Emerson-Lei Automata]
A deterministic transition-based Emerson-Lei automaton (TELA) is a tuple $\mathcal{A} =
(Q, \Sigma, \delta, q_0, \alpha)$
where $Q$ is a finite set of states, $\Sigma$ is an alphabet,  $\delta \colon Q \times \Sigma \rightarrow Q$ is a transition function, $q_0$ is the initial state, and $\alpha \in \mathcal{C}_\delta$ is a generic acceptance condition.  Furthermore we use a superscript to denote a component of a specific automaton, e.g., $\delta^{\mathcal{A}}$ is the transition function of $\mathcal{A}$.
\end{definition}

For convenience we sometimes interpret the transition function as a relation and write $(q, a, q') \in \delta$ instead of $q' = \delta(q,a)$. 
A run $\rho$ of a TELA $\mathcal{A}$ on the $\omega$-word $w$ is an infinite
sequence of transitions $\rho = (q_0, w[0], q_1)(q_1, w[1], q_2)\cdots$
respecting the transition function, i.e. $\rho[i] = (q_i, w[i], q_{i+1}) \in \delta$ for every $i \geq
0$. We denote by $\inf(\rho)$ the set of transitions occurring
infinitely often in the run. A run is called accepting for $\FIN{P}$ if $\inf(\rho) \cap P =
\emptyset$ and accepting for $\INF{P}$ if $\inf(\rho) \cap P \neq \emptyset$.
For arbitrary acceptance conditions $\varphi$, i.e., Boolean combinations of
$\INF{P}$ and $\FIN{P}$, a run is accepting if $\inf(\rho)$ satisfies $\varphi$
in the expected way. 
All well-known acceptance conditions, such as Büchi, Rabin, Streett, and Parity, can be expressed easily using this mechanism.

Since $\INF{P}$ and $\FIN{P}$ are dual, one can complement a deterministic TELA
just by complementing the acceptance condition, i.e., replacing every occurrence
of $\INF{P}$ with $\FIN{P}$, every occurrence of $\FIN{P}$ with $\INF{P}$, and
every disjunction with conjunction and every conjunction with a disjunction.

\section{Construction}

The automaton is constructed from an LTL formula as a product of smaller automata for each temporal subformula. We identified several fragments of LTL in the preliminaries and now present specialised constructions for each of them. While the standard product construction yields an automaton in the size of the product of all automata in the worst-case, the structure of the formula enables us to propagate information, such that we can suspend or disable automata of the product depending on the context.

Consider the following parametric formula: $\G \F (a_1 \wedge \X (a_2 \wedge \dots \X a_m)) \wedge \F (b_1 \wedge \F (b_2 \wedge \dots \F b_n))$. We will later demonstrate that the propagation of information allows us to construct a Büchi automaton of size $O(n + m)$, while \spot{} in the standard configuration yields automata of size $O(n \cdot m)$ and only after enabling simulation-based reductions this decreases to sizes comparable to our automata. Let us now examine the construction, while we translate the formula $\G \F (a_1 \wedge \X a_2) \wedge \F (b_1 \wedge \F b_2)$.

\subsection{Fairness-LTL}
\label{sec:fairness-ltl}

First, we consider the fairness fragment and show that there is a natural way to represent it as deterministic automata. In particular, if we look at Boolean combinations of fairness-LTL formulas ($LTL_{\F \G, \G \F}(\X)$), we  obtain an acceptance condition mirroring the structure of the input formula. Furthermore, if the formula does not contain any $\X$, the automata we obtain is a single-state automaton. For all other formulas we need to store a bounded history in the form of a FIFO-buffer of seen sets of atomic propositions (or valuations). We will now establish the tools necessary to compute the structure of such a buffer. We use the following operations defined on finite and infinite sequences of sets (assuming $n \leq m$):

\[\begin{array}{rclcl}
\text{Pointwise Intersection:} & ~ & u[0] u[1] \dots \historyinter v[0] \dots v[m] & = & (u[0] \cap v[0]) \dots (u[m] \cap v[m]) \emptyset^\omega \\
\text{Pointwise Union:} & ~  & u[0] \dots u[n] \historyunion v[0] \dots v[m] & = & (u[0] \cup v[0]) \dots (u[n] \cup v[n]) \dots v[m] \\
\text{Forward Closure:} & ~ &\cl{w[0] \dots w[n]} & = & w[0] (w[0] \cup w[1]) \dots \bigcup_{k = 0}^n w[k] \\
\text{Drop Last Set of Letters:} & ~ & \drop{w[0] \dots w[n] w[n+1]} & = & w[0] \dots w[n] \\
\end{array}\]

\paragraph{Relevant History.}
Let us consider our example formula: $\G \F (a_1 \wedge \X a_2)$. In order to check whether $w \models a_1 \wedge \X a_2$ holds we just need to know whether $a_1 \in w[0]$ and $a_2 \in w[1]$ holds. The rest of the $w$ can be projected away. The \textit{relevant history} $\mathcal{H}(\varphi)$ for an LTL formula $\varphi$ is a finite word over $2^{AP}$ and masks all propositions that are irrelevant for evaluating $\varphi$. We compute the relevant history $\mathcal{H}$ recursively from the structure of the formula: 
\[\mathcal{H}  \colon LTL(\X) \rightarrow (2^{Ap})^*\]
\[\begin{array}[t]{rclcrcl}
\history{\true} & = & \epsilon & \qquad & \history{\false} & = & \epsilon \\
\history{a} & = & \{a\} & \qquad & \history{\neg a} & = & \{a\} \\
\history{\varphi \wedge \psi} & = &  \history{\varphi} ~ \historyunion ~ \history{\psi} & \qquad & \history{\varphi \vee \psi} & = & \history{\varphi} ~ \historyunion ~ \history{\psi}
\end{array} \qquad \history{\X \varphi} = \emptyset\history{\varphi}\]

\begin{lemma}\label{lem:relevant-history}
Let $\varphi$ be an $LTL(\X)$ formula and let $w$ be a $\omega$-word. Then $w \models \varphi$ if and only if $w ~ \historyinter ~ \history{\varphi} \models \varphi$.
\end{lemma}

\begin{proof}
By induction on $\varphi$. For succinctness we just exhibit two cases and all other cases are analogous. 

\noindent
\textbf{Case $\varphi = \X \psi$.} Then $w \models \varphi$ iff $w_1 \models \psi$ iff $w_1 \historyinter \history{\psi} \models \psi$ iff $\emptyset (w_1 \historyinter \history{\psi}) \models \varphi$ iff $w \historyinter \history{\varphi} \models \varphi$.

\noindent
\textbf{Case $\varphi = \psi \wedge \psi'$.} Then $w \models \varphi$ iff $w \models \psi \wedge w \models \psi'$ iff $w \historyinter \history{\psi} \models \psi \wedge w \historyinter \history{\psi'} \models \psi'$ iff $w \historyinter (\history{\psi} \historyunion \history{\psi'}) \models \varphi$ iff $w \historyinter \history{\varphi} \models \varphi$.
\end{proof}

The TELA we are constructing keeps a buffer masked by $\mathcal{H}$. Intuitively the automaton delays the decision whether $\varphi$ holds by $n = | \history{\varphi} | - 1$ steps and then decides whether it holds true, instead of non-deterministically guessing the future and verifying this guess as done in standard LTL translations.

\begin{definition}
Let $\varphi$ be an $LTL(\X)$ formula over $Ap$ and let $n = \mathsf{max}(|\history{\varphi}| - 1, 0)$. We then define one TELA for $\G\F \varphi$:
\[\automaton{\G\F \varphi} = (Q, 2^{Ap}, \delta, \emptyset^n, \INF{\alpha})\]
\[\begin{array}{lcl}
Q & = & \{w \in (2^{Ap})^n \mid \forall i. ~ w[i] \subseteq \cl{\history{\varphi}}[i]\} \\
\delta(\nu w, \nu') & = & w\nu' ~ \historyinter ~ \drop{\cl{\history{\varphi}}} \quad \text{for all} ~ \nu, \nu' \in 2^{Ap} ~ \text{and} ~ w \in (2^{Ap})^{n-1} \\
\alpha & = & \{(w, \nu, w') \in \delta \mid w\nu\emptyset^\omega \models \varphi\} \\
\end{array}\]
\end{definition}

Observe that we must take the closure of $\history{\varphi}$ before intersecting with the buffer. Otherwise we might lose information while propagating letters from the back to the front of the buffer. Further, we can always drop the last set of letters of the relevant history, since a transition-based acceptance is used. In the context of state-based acceptance this needs to be also stored in the buffer.

Let us apply this construction to our example: $\G \F (a_1 \wedge \X a_2)$.  First, we get $\history{a_1 \wedge \X a_2} = \{a_1\} \{a_2\}$. Second, since we always drop the last set of letters, we have $\drop{\history{a_1 \wedge \X a_2}} = \{a_1\}$ and $n = 1$. Thus we obtain the TELA automaton shown in \Cref{fig:fairness:example}, which is in fact a Büchi automaton.
 
\begin{theorem} 
Let $\varphi$ be an $LTL(\X)$ formula over $Ap$.
\[L(\G\F\varphi) = L(\automaton{\G\F \varphi})\]
\end{theorem}

\begin{proof}
Assume $w \models \G\F\varphi$ holds. Thus we have $\exists^\infty i. ~ w_i \models \varphi$ and we obtain $\exists^\infty i. ~ w_i ~ \historyinter ~ \history{\varphi} \models \varphi$ by using \Cref{lem:relevant-history}. Thus there exists a finite word $w' \in 2^{Ap}$ with (1) $w'\emptyset^\omega = w_i ~ \historyinter ~ \history{\varphi}$, (2) $w'\emptyset^\omega \models \varphi$, and (3) $|w'| = |\history{\varphi}|$. Thus $\automaton{\G\F \varphi}$ infinitely often takes the (shortened) transition $t = (w'[0]\dots w'[n - 1], w'[n])$. Due to (2) we have $t \in \alpha$ and thus $w \in L(\automaton{\G\F\varphi})$. The other direction is analogous.
\end{proof}

Since $\F \G$ is equivalent to $\neg \G \F \neg \varphi$, we immediately obtain also a translation for $LTL_{\F\G}(\X)$. We only need to change the acceptance condition to $\FIN{\alpha}$ with $\alpha = \{(w, \nu, w') \in \delta \mid w\nu\emptyset^\omega \not \models \varphi\}$.

\begin{figure}
	\centering
	\begin{subfigure}[b]{0.45\textwidth}
		\centering
		\begin{tikzpicture}[x=2cm,y=2cm,font=\footnotesize,initial text=]
		\node[state,initial] (1) at (0,0)      {$\{\}$};
		\node[state]          (2) at (0,-1)     {$\{a_1\}$};
		\path[->] (1) edge node[right]{$a_1$} (2)
                   		   edge [loop right] node[right]{$\overline{a_1}$} (1)
              		      (2) edge [bend left=-45] node[right] {$\overline{a_1}  \overline{a_2}$} (1)
                   		   edge [bend left=45, double] node[left] {$\overline{a_1} a_2$} (1)
                  	           edge [loop left, double] node[left]{$a_1 a_2$} (2)
                   		   edge [loop right] node[right]{$a_1 \overline{a_2}$} (2);
		\end{tikzpicture}
		\caption{$\psi_1 = \G\F (a_1 \wedge \X a_2)$}
		\label{fig:fairness:example}
	\end{subfigure}
	\begin{subfigure}[b]{0.45\textwidth}
	\centering
	\begin{tikzpicture}[x=2cm,y=2cm,font=\footnotesize,initial text=]
\node[state,initial] (1) at (0,0)      {$\F (b_1 \wedge \F b_2)$};
\node[state]          (2) at (1,0)     {$\F b_2$};
\node[state] (3) at (2,0)     {$\true$};

\path[->] (1) edge node[below]{$b_1 \overline{b_2}$} (2)
 	           edge[bend left=-45] node[below]{$b_1 b_2$} (3)
                   edge [loop above] node[above]{$\overline{b_1}$} (1)
              (2) edge node[below] {$b_2$} (3)
                    edge [loop above] node[above]{$\overline{b_2}$} (2)
              (3) edge [loop above, double] node[above]{$\true$} (3);
\end{tikzpicture}
\caption{$\psi_2 = \F (b_1 \wedge \F b_2)$}
\label{fig:cosafety:example}
	\end{subfigure}
\caption{Automata for $\psi_1$ and $\psi_2$. Bold edges denote accepting transitions.}
\end{figure}
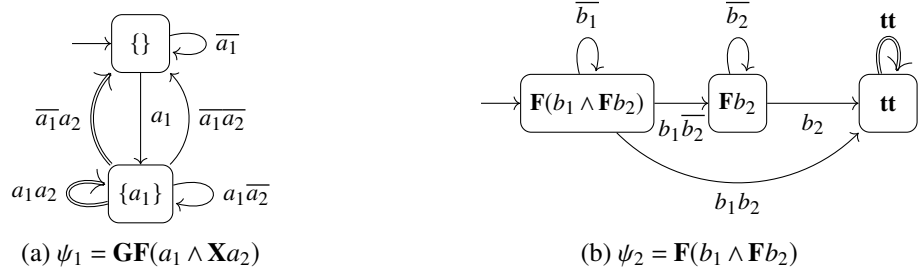

\subsection{Safety- and Cosafety-LTL}
\label{sec:safety}

Translating safety LTL to deterministic automata is a well-studied problem. Since these languages can be defined using bad prefixes, meaning once a bad prefix has been read, the word is rejected, most automata generated by most available translations will have a single rejecting sink. All other states and transitions are then either rejecting or accepting. We use the straight-forward approach to apply the $\aft$-function from \cite{EsparzaKS16} to obtain a deterministic automaton for cosafety LTL formulas and by duality also for automata for safety languages. The $\aft$-function computes the left-derivative of a language expressed as an LTL formula.  

\begin{definition}[\cite{EsparzaKS16}, Definition 7]
Let $\varphi$ be a formula of $LTL(\U, \X)$, then \[\automaton{\varphi} = (Q, 2^{Ap}, \delta, [\varphi]_P, \{[\true]_P\}).\]
\end{definition}

\begin{theorem}[\cite{EsparzaKS16}, Theorem 2]
Let $\varphi$ be a formula of $LTL(\U, \X)$, then \[L(\varphi) = L(\automaton{\varphi}).\]
\end{theorem}

For the cosafety formula $\F (b_1 \wedge \F b_2)$ we then obtain the automaton of \Cref{fig:cosafety:example} with the accepting sink $[\true]_P$. This approach also immediately tells us, when a run is accepting by looking at the state.

\subsection{General LTL}
\label{sec:general-ltl}

If the translation encounters a subformula not covered by \Cref{sec:fairness-ltl} and \Cref{sec:safety}, it resorts to an external general purpose LTL to deterministic automaton translation. Here no restrictions on the type of the automaton are made, since all of them --- Rabin, Streett, Parity, Büchi --- can be interpreted as a TELA. 

\subsection{Product construction}

\paragraph{Standard Construction.} All these deterministic automata are then combined using a product construction. We first introduce the standard product construction for Emerson-Lei Automata that is similar to the product construction for Muller automata and then move on to the enhanced construction.

\begin{definition}\label{def:naive-product}
    Let $\varphi$ be a formula and for every $\psi \in \sff{\varphi}$ let $\automaton{\psi}$ be a deterministic TELA recognising $L(\psi)$. The deterministic TELA for the product automaton is defined as:
\[\prodautomaton{\varphi} = (Q, 2^{Ap}, \delta, q_0, \alpha(\varphi))\]
\[\delta(s, \nu) = \{\psi \mapsto \delta^{\automaton{\psi}}(s[\psi], \nu) \mid \psi \in \sff{\varphi}\} 
\qquad q_0 = \{\psi \mapsto q_0^{\automaton{\psi}} \mid \psi \in \sff{\varphi}\}\]

We denote by $s[\psi] = q$ the current state of the automaton $\automaton{\psi}$ in the product state $s$, meaning $\psi \mapsto q \in s$. Since all $\delta^\automaton{\psi}$ are deterministic, $\delta$ is also deterministic. We denote by $Q^{\automaton{\psi}}$ the states of $\automaton{\psi}$ and by $q_0^{\automaton{\psi}}$ the initial state of $\automaton{\psi}$. Further $Q$ is defined as the set of all from the initial state reachable states. The acceptance condition is recursively computed over the structure of $\varphi$ with $\uparrow$ denoting the lifting of the acceptance condition:
\[\begin{array}{rclcrcl}
\alpha(\true) & = & \true & \qquad & \alpha(\varphi \wedge \psi) & = & \alpha(\varphi) \wedge \alpha(\psi) \\
\alpha(\false) & = & \false & & \alpha(\varphi \vee \psi) & = & \alpha(\varphi) \vee \alpha(\psi) \\
\end{array} \qquad \alpha(\psi) ~~ = ~~ \lift{\alpha^{\automaton{\psi}}} \]
\end{definition}

\begin{theorem} Let $\varphi$ be an LTL formula. Then
\[L(\varphi) = L(\prodautomaton{\varphi})\]
\end{theorem}

\paragraph{Enhanced Construction.} An essential part of the enhanced product construction is the removal of unnecessary information from the product states. For this we introduce three additional states with special semantics: $\qacc$ signalises that the component moved to an accepting trap, while $\qrej$ expresses that the component moved to a rejecting trap. Alternatively, if a component got irrelevant for the acceptance condition it is also moved to $\qrej$. Lastly, $\qhold$ says that the component was put on hold. More specifically, we put the fairness automata on hold, if a \enquote{neighbouring} automaton still needs to fulfil its goal, such as reaching an accepting trap. To make notation easier to read we assume that every automaton $\automaton{\varphi}$ contains these states and all accepting sinks (or traps) have been replaced by $\qacc$ and rejecting by $\qrej$.

\newcommand{\update}{\ensuremath{\mathsf{update}}}
\newcommand{\run}[1]{\ensuremath{\mathsf{run}\left(#1\right)}}
\newcommand{\prune}[1]{\ensuremath{\mathsf{prune}(#1)}}

\noindent In the following we use the following abbreviations to reason about LTL formulas:

\begin{itemize}
\item $\conj{\varphi}$ ($\disj{\varphi}$) denotes the set of all conjuncts of a conjunction (disjuncts of a disjunction) outside the scope of a temporal operator, e.g. let $\varphi = \F a \wedge (\X b \vee \G c)$, then $\conj{\varphi} = \{\{\F a, \X b \vee \G c\}\}$ and $\disj{\varphi} = \{\{\X b, \G c\}\}$.  
\item $\varphi[\Psi/\psi]$ denotes the substitution of all formulas in the set $\Psi$ with the formula $\psi$, e.g. $(\F a \wedge (\X b \vee \G c))[\{\F a, \G a\}/ \true] = \true \wedge (\X b \vee \G c)$.
\item $\support{\varphi}$ denotes the support of a formula, where the formula is viewed as a propositional formula, which means that temporal operators are also considered propositions, e.g. $\support{(\X a \wedge \F b) \vee (\F b)} = \{\F b\}$. This means every assignment can be restricted to the propositions of the support: $S \models_P \varphi \leftrightarrow S \cap \support{\varphi} \models_P \varphi$, where $\models_P$ denotes the conventional propositional satisfaction relation.
\end{itemize}

We use the following definitions to manipulate product states:

\begin{definition}[Product State Modifications] 
An $\mathsf{update}$ of a product state tests a predicate $P$ on a formula-state pair $(\psi, q)$ and replaces $q$ with a new value obtained by the updater $U$ depending on $\psi$, if it holds:
\[\update(s, P, U) = \{\psi \mapsto (\textbf{if } P(\psi, q) \textbf{ then } U(\psi) \textbf{ else } q) \mid \psi \mapsto q \in s \}\]

\noindent $\prune{s}$ disables automata in $s$ that became irrelevant for the acceptance condition, meaning there are no longer in the support of the original formula after using knowledge from other automata. For this let us denote by $\Psi_{\mathsf{acc}}$ all $\psi \mapsto \qacc \in s$ and by $\Psi_{\mathsf{rej}}$ all $\psi \mapsto \qrej \in s$.

\[\begin{array}{ll} \prune{s} & = \update(s, P, U) \\ 
P(\psi, q) & = (q \neq \qacc \wedge \psi \not \in \support{\varphi[\Psi_{\mathsf{acc}}/\true, \Psi_{\mathsf{rej}}/\false]}) \\
U(\psi) & = \qrej
\end{array}\]

\noindent $\run{s}$ starts (fairness) automata that are required for the acceptance but have been put on hold. This is the case, if automata with terminal acceptance for formulas in the same conjunction ($\run{s}_c$) have not yet reached $\qacc$ or the dual case for disjunctions:
\[
\begin{array}{ll}
\run{s}_c & = \update(s, P_c, U) \\ 
\run{s}_d & = \update(s, P_d, U) \\ 
P_c(\psi, q) & = (q = \qhold \wedge \exists C \in \conj{\varphi}. ~ \psi \in C \wedge \forall \chi \in C \cap LTL(\U,\X). ~ s[\chi] = \qacc) \\
P_d(\psi, q) & = (q = \qhold \wedge \exists D \in \disj{\varphi}. ~ \psi \in D \wedge \forall \chi \in D \cap LTL(\R,\X). ~ s[\chi] = \qrej) \\
U(\psi) & = q_0^{\automaton{\psi}}
\end{array}\]
\end{definition}

\begin{definition}[Enhanced Product Automaton]
Let $\varphi$ be a formula. The TELA for the enhanced product automaton is defined the same way as  \Cref{def:naive-product} with the following changes:

\[\mathcal{A}^\times_E(\varphi) = (Q, 2^{Ap}, \delta, q_0, \alpha(\varphi))\]
\[\begin{array}{lcl}
\delta(s, \nu) & = & \run{\prune{\{ \psi \mapsto \delta^{\automaton{\psi}}(s[\psi], \nu) \mid \psi \in \sff{\varphi}\}}} \\[1em]
q_0 & = & \run{\left\{\psi \mapsto \begin{cases}	 q_0^{\automaton{\psi}} & \text{if} ~~ \psi \in \sff{\varphi} \setminus LTL_{\F\G, \G\F}(\X)\\ \qhold & \text{otherwise} \end{cases}\right\}} \\
\end{array}\]

\end{definition}

\begin{theorem} Let $\varphi$ be a formula.
\[L(\varphi) = L(\mathcal{A}^\times_E(\varphi))\]
\end{theorem}

\noindent If we apply this construction to $\G \F (a_1 \wedge \X a_2) \wedge \F (b_1 \wedge \F b_2)$ we obtain the automaton shown in \Cref{fig:suspended:example}. Observe that $\psi_2$ is put on hold until the automaton for $\psi_1$ reaches $\qacc$.
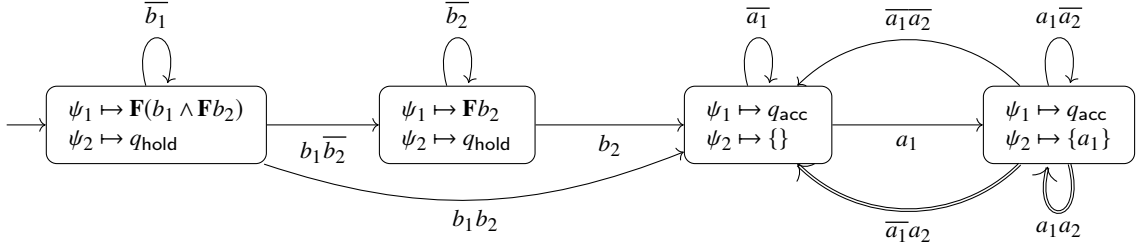
\begin{figure}
\begin{center}
\begin{tikzpicture}[x=2cm,y=2cm,font=\footnotesize,initial text=]
\node[state,initial] (1) at (0,0)      {$\begin{array}{l} \psi_1 \mapsto \F (b_1 \wedge \F b_2) \\ \psi_2 \mapsto \qhold \end{array}$};
\node[state]          (2) at (2,0)   {$\begin{array}{l} \psi_1 \mapsto \F b_2 \\ \psi_2 \mapsto \qhold \end{array}$};
\node[state]          (3) at (4,0)     {$\begin{array}{l} \psi_1 \mapsto \qacc \\ \psi_2 \mapsto \{\} \end{array}$};
\node[state]          (4) at (6,0)   {$\begin{array}{l} \psi_1 \mapsto \qacc \\ \psi_2 \mapsto \{a_1\} \end{array}$};

\path[->] (1) edge node[below]{$b_1 \overline{b_2}$} (2)
 	           edge[bend left=-20] node[below]{$b_1 b_2$} (3)
                   edge [loop above] node[above]{$\overline{b_1}$} (1)
              (2) edge node[below] {$b_2$} (3)
                    edge [loop above] node[above]{$\overline{b_2}$} (2);
                    
\path[->] (3) edge node[below]{$a_1$} (4)
                   edge [loop above] node[above]{$\overline{a_1}$} (3)
              (4) edge [bend left=-45] node[above] {$\overline{a_1}  \overline{a_2}$} (3)
                    edge [bend left=45, double] node[below] {$\overline{a_1} a_2$} (3)
                   edge [loop below, double] node[below]{$a_1 a_2$} (4)
                   edge [loop above] node[above]{$a_1 \overline{a_2}$} (4);
\end{tikzpicture}
\caption{Enhanced Product Automaton for $\G \F (a_1 \wedge \X a_2) \wedge \F (b_1 \wedge \F b_2)$, only the accepting edges for $\psi_2$ are drawn.}
\label{fig:suspended:example}
\end{center}
\end{figure}

\subsubsection{Further Optimisations}

There are two further optimisations we implement: First, we replace the local histories of each automaton for $LTL_{\F\G,\G\F}(\X)$ with one \textit{global} history. Second, we \textit{piggyback} the acceptance of (co-)safety automata on neighbouring fairness automata. Let $C \in \conj{\varphi}$ be a conjunction, $\psi_r \in LTL(\U, \X) \cap C$ and $\psi_f \in LTL_{\F\G}(\X)\cap C$. We then have $\alpha^{\automaton{\psi_f}} = \FIN{S}$ and extend $S$ with $Q^{\automaton{\psi_r}} \setminus \{\qacc\}$. The same trick can be applied to $\psi_f \in LTL_{\G\F}(\X)$ and of course to the dual case with $\psi_s \in LTL(\R, \X)$.

\section{Succinctness}
\label{succinctness}

It is clear from \Cref{def:naive-product} that the presented translation uses at most $|\sff{\varphi}|$ acceptance sets for Boolean combinations of $LTL_{\F\G,\G\F}(\X)$. We show succinctness compared to deterministic generalized Rabin automata or deterministic Streett automata, which might need an exponential sized acceptance condition for the same language, while the acceptance size only grows linearly for TELAs.

For this, we define two mutually recursive formula patterns modelling Rabin and Streett conditions:
\begin{align*}
    \varphi_{\mathrm{R}, 0} &= \F \G a_0 \wedge \G \F b_0 \qquad \varphi_{\mathrm{R}, n+1} = (\F \G a_{n+1} \wedge \G \F b_{n+1}) \vee \varphi_{\mathrm{S}, n}\\
    \varphi_{\mathrm{S}, 0} &= \F \G a_0 \vee \G \F b_0 \qquad \varphi_{\mathrm{S}, n+1} = (\F \G a_{n+1} \vee \G \F b_{n+1}) \wedge \varphi_{\mathrm{R}, n}
\end{align*}

We call the subformulas $\F \G a_j$, $\G \F b_j$ leafs, and a set $\mathcal{L} \subseteq \sff{\varphi}$
of leafs a good leaf set --- denoted by $\gls(\varphi)$ ---, if it is a minimal set satisfying
$\varphi_{\mathrm{R}, n}$, respectively $\varphi_{\mathrm{S}, n}$. 

\begin{lemma}
    \label{good-leaf-sets}
    For $\varphi_{\mathrm{R}, n}$ and $\varphi_{\mathrm{S}, n}$
    there are $\varOmega(2^{\frac{n}{2}})$ good leaf sets.
\end{lemma}

\begin{proof}
    First note, that for each subformula of the form $\varphi_{\mathrm{S}, k+1}$
    there is a doubling of good leaf sets in $\varphi_{\mathrm{R}, k}$.
    This comes from the conjunction of the Streett pair $\F \G a_{k+1} \vee \G \F
    b_k$ and
    $\varphi_{\mathrm{R}, k}$: To every good leaf set of
    $\varphi_{\mathrm{R}, k}$ one has to add either $\F \G a_{k+1}$ or $\G \F
    b_{k+1}$ to obtain a good leaf set for $\varphi_{\mathrm{S}, k+1}$. On the
    other side, $\gls(\varphi_{\mathrm{R}, k+1}) = \left\lbrace \F \G a_{k+1},
    \G \F b_{k+1} \right\rbrace \cup \gls(\varphi_{\mathrm{S},n})$ 

    Since we alternate between $\varphi_{\mathrm{R}, k}$ and
    $\varphi_{\mathrm{S}, k}$, we have $\varTheta(2^{\frac{n}{2}})$ good
    leaf sets for $\varphi_{\mathrm{R}, n}$ (resp.
    $\varphi_{\mathrm{S}, n}$).
\end{proof}

W.l.o.g. we assume, that every good leaf set contains at most one subformula of
the form $\F \G \psi$. If there are more than one subformulas of this pattern,
e.g. $\F \G \psi_1$ and $\F \G \psi_2$, one can remove both and add $\F \G
(\psi_1 \wedge \psi_2)$. Note that this transformation does not reduce the
number of good leaf sets, since no good leaf set is removed, and two good leaf
sets cannot be reduced to the same good leaf set.

One can easily give a bijection from a Rabin pair to a good leaf set, negate
$\varphi_{\mathrm{S},n}$ to $\varphi_{\mathrm{R}, n}$, and use the duality
between Rabin and Streett automata.  Overall, we get the following lemma:

\begin{lemma}
    \label{acc-genRabin}
    For every $n \in \mathbb{N}$, every generalized Rabin automaton equivalent to
    $\varphi_{\mathrm{S}, n}$ has at least $\vert
    \gls(\varphi_{\mathrm{S}, n}) \vert$ acceptance pairs. 
    For every $n \in \mathbb{N}$, every Streett automaton equivalent to
    $\varphi_{\mathrm{R}, n}$ has at least $\vert
    \gls(\varphi_{\mathrm{S}, n}) \vert$ acceptance pairs. 
\end{lemma}

Note that $\varphi_{\mathrm{S}, n}$ and $\varphi_{\mathrm{R}, n}$
are Boolean combinations of formulas from $LTL_{\F\G,\G\F}(\X)$. Since
$\varphi_{\mathrm{S}, n}$ and $\varphi_{\mathrm{R}, n}$ do not
contain a $\X$ operator, the produced automaton of our construction has exactly
one state. According to \Cref{sec:fairness-ltl} one can see, that the structure
of the formula is directly translated into the acceptance condition. Therefore,
the length of the acceptance condition is equal to $\vert
\varphi_{\mathrm{S}, n} \vert$ (resp. $\vert \varphi_{\mathrm{R}, n}
\vert$).

\section{Experimental Evaluation}

Our experimental evaluation is two-part: At first, we evaluate our translation
by comparing the automata sizes and acceptance sizes. The second contribution in
our evaluation considers probabilistic model checking with the help of automata.
For every experiment, we set a time limit of $30$ minutes and a memory limit of
$10$ GB for every process.\footnote{%
All experiments were carried out on a computer with
two
Intel E5-2680 8-core CPUs at 2.70~GHz with 384GB of RAM running Linux.%
}

\subsection{Automata Sizes}

For the comparison of the acceptance conditions, we rely on counting the number of
$\FIN{\cdot}$ and $\INF{\cdot}$ occurring in the acceptance condition. We
compare our tool \delag{} with \rabinizer{}
\cite{EsparzaKS16} and \ltltotgba{} of \spot{}. %
Our benchmark consists of $94$ LTL formulas from
\cite{SomBloem00, DwyerAC99, EH00} where for $34$ formulas \delag{} was
able to translate a formula completely without using an external tool.   For these
formulas we do not need to rely on an external tool translating LTL to
deterministic automata. Should we require external tools to translate parts of
the formula, as described in \Cref{sec:general-ltl}, we  use \ltltotgba{} of
\spot{} as fallback solution.

Overall, \delag{} produced automata with a minimal state space in $77$ cases,
followed by \ltltotgba{} with $71$ formulas. For the comparison of the
acceptance, \delag{} has delivered the smallest acceptance for $59$ formulas,
whereas \ltltotgba{} could produce an automaton with a minimal acceptance
condition for $56$ formulas. As it can be seen in \Cref{table:states-benchmark}
\delag{}, \ltltotgba{} and \rabinizer{} show roughly the same behavior,
generating for $36$ vs. $37$ vs. $35$ formulas automata with size less or equal
than $3$, with a slight advantage for \delag{} producing more automata of size
one.
\begin{table*}[htbp]
\caption{Overview of the number of automata generated by the tools
\delag{}, \ltltotgba{}, \rabinizer{} with an upper bound of states (on the left side)
and with an upper bound of the number of leafs in the acceptance condition.}
\label{table:states-benchmark}
\centering
\begin{minipage}{0.49\textwidth}%
\begin{tabular}{l|rrrrrrr}
{\#}States $\leqslant x$ & $  1$ & $  2$ & $  3$ & $  4$ & $  6$ & $10$ &\hspace{-1.5 ex}$ > 10$\\\hline
\delag{} & $9$ & $17$ & $36$ & $59$ & $75$ & $87$ & $7$\\
\ltltotgba{} & $6$ & $17$ & $37$ & $60$ & $78$ & $89$ & $5$\\
\rabinizer{} & $6$ & $15$ & $35$ & $53$ & $75$ & $84$ & $10$\\
\end{tabular}%
\end{minipage}%
\begin{minipage}{0.49\textwidth}%
\hspace{2 ex}
\begin{tabular}{l|rrrrrr}
    Acc. size $\leqslant x$ & $  1$ & $  2$ & $  3$ & $  4$ & $  6$ &\hspace{-1 ex}$> 6$\\\hline
\delag{} & $50$ & $79$ & $83$ & $83$ & $90$ & $4$\\
\ltltotgba{} & $72$ & $84$ & $84$ & $86$ & $93$ & $1$\\
\rabinizer{} & $20$ & $34$ & $54$ & $67$ & $81$ & $13$\\
\end{tabular}
\end{minipage}
\end{table*}

The situation differs for the sizes of the acceptance condition: \ltltotgba{} generates $72$ automata with 
acceptance size  $1$ whereas \delag{} generates $50$ automata with
acceptance size $1$. For bigger acceptance sizes the number of generated
automata are similar for \ltltotgba{} and \delag{}. In comparison,
\rabinizer{} tends to produce automata with bigger acceptance sizes.

For the formulas $\varphi_{\mathrm{R},n}$ of \Cref{succinctness} the results are
as expected (see \Cref{table:acc-alternating}). \delag{} always produces the
smallest acceptance with a one state automaton, whereas the acceptance sizes of
the automata produced by \rabinizer{} grow faster, e.g., for $n=5$ and $n=7$
\rabinizer{} produces an automaton with acceptance size $45$ and $109$,
respectively. Both \delag{} and \rabinizer{} produce one state automata.
\ltltotgba{} behaves differently: The state space size of the automata grows
with $n$: for $n=1$ \ltltotgba{} produces an automaton with $7$ states and an
acceptance size of $4$, whereas for $n=3$ the state space increased to $21889$
states and an acceptance size of $20$. For $n > 3$ we were not able to produce
automata with \ltltotgba{}.

\begin{table*}[htbp]
\renewcommand{\arraystretch}{1.3}
\caption{Acceptance sizes for the alternating formula $\varphi_{\mathsf{R}, n}$;
$-$ means time-out or mem-out.}
\label{table:acc-alternating}
\centering
\begin{tabular}{r|rrrrrrrr}
$n = $ & $0$ & $1$ & $2$ & $3$ & $4$ & $5$ & $6$ & $7$\\\hline
\delag{} & $2$ & $4$ & $6$ & $8$ & $10$ & $12$ & $14$ & $16$\\
\ltltotgba{} & $2$ & $4$ & $8$ & $20$ & $-$ & $-$ & $-$ & $-$\\
\rabinizer{} & $2$ & $5$ & $7$ & $17$ & $19$ & $45$ & $47$ & $109$%
\end{tabular}
\end{table*}

For the evaluation of the history, we took the formula pattern
$\varphi_{\mathcal{H},n}$:

$$ \varphi_{\mathcal{H}, n} = \begin{cases}
    \bigl(\F{} \G{} (a \vee \X^n b)\bigr) \vee \varphi_{\mathcal{H}, n-1} &
    \text{if } n \textrm{ is even}\\
    \bigl(\F{} \G{} (\neg a \vee \X^n b)\bigr) \vee \varphi_{\mathcal{H}, n-1} & \text{otherwise}
\end{cases}$$

Every subformula $a \vee \X^n b$ (or $\neg a \vee \X^n b$)
commits the first position or the $n$-th position. So only two out of $n$
positions may be fixed, and hence we can share a lot of the state space
between the $\F{} \G{}$ formulas. 

The results can be found in \Cref{table:history-vs-after}. The state space of
\ltltotgba{} grows faster than \delag{}, the former being only capable to produce automata
up to $n=5$ before hitting the memory limit. For \rabinizer{}, we were not able
to produce automata for $n \geqslant 4$, since \rabinizer{} supports only a
limited number of acceptance set. This shows, that the acceptance condition
grows immensely.

\begin{table*}[htbp]
\renewcommand{\arraystretch}{1.3}
\caption{Automata sizes and number of acceptance sets for $\varphi_{\mathcal{H},n}$;
$-$ means time-out or mem-out.}
\label{table:history-vs-after}
\centering
\begin{tabular}{r|r|rrrrrrrr}
    \multicolumn{2}{r|}{$n = $} & $0$ & $1$ & $2$ & $3$ & $4$ & $5$ & $6$ & $7$\\\hline
    \multirow{2}{*}{\delag{}} & {\#}States & $1$ & $2$ & $4$ & $8$ & $16$ & $32$ & $64$ & $128$\\
    & Acc. size& $1$ & $2$ & $3$ & $4$ & $5$ & $6$ & $7$ & $8$\\\hline
    \multirow{2}{*}{\ltltotgba{}} & {\#}States & $2$ & $4$ & $21$ & $170$ & $1816$ & $22196$ & $-$ & $-$\\
     & Acc. size & $2$ & $2$ & $2$ & $2$ & $2$ & $2$ & $-$ & $-$\\\hline
    \multirow{2}{*}{\rabinizer{}} & {\#}States & $1$ & $2$ & $5$ & $11$ & $-$ & $-$ & $-$ & $-$\\
    & Acc. size & $1$ & $3$ & $7$ & $19$ & $-$ & $-$ & $-$ & $-$
\end{tabular}
\end{table*}

\subsection{Prism Runtimes}
\label{sec:prism}

We have implemented a routine for the analysis of MDPs in \prism. Here we
compare the behaviour of \prism{} if the three tools \delag{}, \ltltotgba{} from
\spot, and \rabinizer{} are employed as automata generation tools. As case study
we consider the IEEE 802.11 Wireless LAN Handshaking protocol. It describes a
resolving mechanism to stop interference if two stations want to send a message
at the same time. The key trick is, that all participating stations listen to
interference, and if a message has become garbled, the stations waits a random
amount of time (limited by an upper bound called \texttt{Backoff}) and then
tries to resend the message. We used the following properties:

\begin{itemize}
    \item ``If a message from sender $i$ has been garbled, it will be sent
        correctly in the future''\\
        $\varphi_1 = \bigwedge_{1 \leq i \leq n} \G{} \, (\texttt{garbled}_i \rightarrow \F{}
        \, \texttt{correct}_i)$
    \item ``Every sender sends at least one message correctly.'' : $\varphi_2 = \bigwedge_{1 \leq i \leq n} \F{} \, \texttt{correct}_i$
    \item ``The first time every station wants to send, the channel remains
        free for $k$ steps''\\
        $\varphi_3 = \bigwedge_{1 \leq i \leq n} \texttt{wait}_i \, \U \,  (\texttt{wait}_i \wedge
       \G^{\leq k} \, \texttt{free})$
       where $\G^{\leq k} \texttt{free} = \texttt{free} \wedge \X{} \,
       \texttt{free} \wedge \ldots \wedge \X^n \, \texttt{free}$
   \item ``Every station, that wants to send a message infinitely often, is
       able to send a message correctly infinitely often '' : $\varphi_4 = \bigwedge_{1 \leq i \leq n} (\G{} \, \F{} \,
       \texttt{wait}_i) \rightarrow  (\G{} \, \F{} \, \texttt{correct}_i)$
    \item ``Every station satisfies both the reachability formula $\varphi_2$
        and the fairness formula $\varphi_4$''\\
        $\varphi_5 = \bigl( \bigwedge_{1 \leq i \leq n} \F{} \,
        \texttt{correct}_i \bigr) \wedge \bigl(\bigwedge_{1 \leq i \leq n} (\G{} \, \F{} \,
       \texttt{wait}_i) \rightarrow  (\G{} \, \F{} \, \texttt{correct}_i) \bigr)$
\end{itemize}

Every property can be translated directly by \delag{} without external tools,
except $\varphi_1$, for which we translate the subformulas $\G{} \,
(\texttt{garbled}_i \rightarrow \F{} \, \texttt{correct}_i)$ with \ltltotgba{}
and then build the product. So $\varphi_1$ should be seen as a benchmark for the
product construction.

For all properties we asked for the minimal ($\pmin{\cdot}$) or maximal
($\pmax{\cdot}$) probability of the IEEE 802.11 handshaking model with two
stations and a \texttt{Backoff} of at most $3$ to satisfy the property. If a
formula has a window length (e.g. $\G^{\leq k}$) we uniformly choose $k=6$.
\Cref{table:prism} lists some measured time values and automata/product sizes.
All \prism{} experiments were carried out with the hybrid engine, an engine that
combines symbolic and explicit data structures offering a good compromise.

\begin{table*}[tbp]
\renewcommand{\arraystretch}{1.3}
\caption{\prism{} runtimes ($t_{MC}$) for the IEEE 802.11 case study enhanced with automata
sizes ($\vert A \vert$) and the number of BDD nodes in the product (BDD
size $\mathcal{M} \otimes \mathcal{A}$)}
\label{table:prism}
\centering
\begin{tabular}{r|rrr|rrr|rrr}
\multirow{3}{*}{Property} &
\multicolumn{3}{c|}{\delag} &
\multicolumn{3}{c|}{\ltltotgba} &
\multicolumn{3}{c}{\rabinizer}\\\cline{2-10}
&
\multirow{2}{*}{$\vert \mathcal{A} \vert$} &
\multicolumn{1}{c}{BDD size} &
\multirow{2}{*}{$t_{MC}$} &
\multirow{2}{*}{$\vert \mathcal{A} \vert$} &
\multicolumn{1}{c}{BDD size} &
\multirow{2}{*}{$t_{MC}$} &
\multirow{2}{*}{$\vert \mathcal{A} \vert$} &
\multicolumn{1}{c}{BDD size} &
\multirow{2}{*}{$t_{MC}$}\\
&
&
\multicolumn{1}{c}{$\mathcal{M} \otimes \mathcal{A}$} &
&
&
\multicolumn{1}{c}{$\mathcal{M} \otimes \mathcal{A}$} &
&
&
\multicolumn{1}{c}{$\mathcal{M} \otimes \mathcal{A}$} &
\\\hline
$\pmin{\varphi_{1}}$ & \pnodes{4} & \pnodes{31861} & \psec{6.583} & \pnodes{5} & \pnodes{44181} & \psec{9.507} & \pnodes{4} & \pnodes{31861} & \psec{32.222}\\
$\pmin{\varphi_{2}}$ & \pnodes{4} & \pnodes{61711} & \psec{165.448} & \pnodes{4} & \pnodes{61719} & \psec{160.624} & \pnodes{4} & \pnodes{61719} & \psec{158.991}\\
$\pmin{\varphi_{3}}$ & \pnodes{20} & \pnodes{46013} & \psec{27.455} & \pnodes{20} & \pnodes{46106} & \psec{26.23} & \pnodes{72} & \pnodes{47114} & \psec{27.974}\\
$\pmin{\varphi_{4}}$ & \pnodes{1} & \pnodes{30091} & \psec{42.594} & \pnodes{5} & \pnodes{30473} & \psec{6.798} & \pnodes{1} & \pnodes{30091} & \psec{47.017}\\
$\pmax{\varphi_{4}}$ & \pnodes{1} & \pnodes{30091} & \psec{5.732} & \pnodes{32} & \pnodes{129905} & \psec{273.79} & \pnodes{1} & \pnodes{30091} & \psec{6.026}\\
$\pmin{\varphi_{5}}$ & \pnodes{4} & \pnodes{61711} & \psec{120.856} & \pnodes{21} & \pnodes{65504} & \psec{91.646} & \pnodes{4} & \pnodes{61719} & \psec{125.519}\\
$\pmax{\varphi_{5}}$ & \pnodes{4} & \pnodes{61711} & \psec{152.727} & \pnodes{40} & \pnodes{182133} & \psec{861.603} & \pnodes{4} & \pnodes{61719} & \psec{165.061}\\
\end{tabular}
\end{table*}

First, the generation time for every automaton was below \psec{1}, except for
\rabinizer{} at $\pmin{\varphi_3}$ where it was \psec{1.844}. In $3$ cases
\prism{} in combination with \delag{} was the fastest. For $\pmin{\varphi_4}$
\ltltotgba{} took only \psec{6.798} in comparison to \psec{42.594} for \delag{}
despite the smaller automaton, since one heuristic applied for \ltltotgba{} that
did not apply for \delag{}: For the analysis of maximal end-components (MEC) we checked always at first, if
the whole MEC satisfies the acceptance condition, and only if not, we look for
accepting sub-end-components within the MEC. For \ltltotgba{} the whole MEC was
accepting, but for \delag{} one had to search for an accepting sub-end-component. Since in a symbolic representation SCC enumeration is costly,
\ltltotgba{} was much faster.

In general, one can see, that \delag{} produced every time the smallest
automaton, that also results in the smallest number of BDD nodes in the product
and comparatively small model checking times.

We have checked two more properties in the full version \cite{gandalf17web} as well as included a
comparison with the standard approach of \prism{} that uses an own
implementation of \texttt{ltl2ba} \cite{GO01} and Safra's determinisation
\cite{Safra88} and delivers automata with state-based acceptance.

\section{Conclusion}

We presented a general framework based on the product construction and
specialised translations for fragments of LTL to build deterministic Emerson-Lei automata. In
particular, for the important fairness fragment we established an efficient
construction, where the state space only depends on the nesting
depth of $\X$, and all of the complexity is shifted to the acceptance
condition. The general construction applies a range of additional optimisations: such as pushing temporal
operators down the syntax tree, piggybacking to reduce the number of 
acceptance sets and sharing of equal automata parts. In particular our history buffer 
approach reduces the state space,
since the buffer can be shared between automata for different
subformulas. If a formula does not belong to one of our explicitly supported fragments, 
we can run an external LTL to deterministic automaton
translator and incorporate the resulting automaton via product construction and
lifting.

Benchmarking this approach has shown the potential of our method. Standard
benchmarks highlight the potential of allowing more complex acceptance conditions,
our tool had a slight advantage in the state space over \spot. Those results also reflect in the area
of probabilistic model checking, where we analysed the IEEE 802.11 Handshaking
protocol.

However, the heuristics presented here are not complete, and this approach
should be understood as a framework. So, one direction for future work is to add
more explicitly supported LTL fragments. Another point would be to analyse the subformulas, which cannot be
translated directly and choose an external tool, that behaves well for these
specific subformulas. For example, it is well-known, that obligation LTL
formulas can be translated to weak DBA, and then efficiently minimised. This is
implemented in \spot. Another direction one could take a deeper look into, is to
start with a non-deterministic Büchi automaton, and try to find small
deterministic automaton with a complex acceptance condition. Of course, general
methods to shrink the state space like bisimulation could be also applied. Also,
the particular ingredients of our transformation could optimised further, e.g.
the history could be allocated dynamically, and therefore reduce the state space
even further without increasing the acceptance condition complexity.

\paragraph{Acknowledgments.} The authors want to thank the anonymous reviewers for the constructive feedback.

\bibliographystyle{eptcs}
\bibliography{lit}
\end{document}